\documentclass[11pt]{article}

\newcommand{\comment}[1]{}

\usepackage{epsfig, epstopdf}
\usepackage{amsmath, amsthm, amsfonts, amssymb, latexsym, stmaryrd}
\newtheorem{lemma}{Lemma}[section]
\newtheorem{theorem}[lemma]{Theorem}
\newtheorem{corollary}[lemma]{Corollary}

\newtheorem{observation}[lemma]{Observation}
\newtheorem{definition}[lemma]{Definition}
\newtheorem{problem}[lemma]{Problem}

\def\squareforqed{\hbox{\rlap{$\sqcap$}$\sqcup$}}
\def\qed{\ifmmode\squareforqed\else{\unskip\nobreak\hfil
\penalty50\hskip1em\null\nobreak\hfil\squareforqed
\parfillskip=0pt\finalhyphendemerits=0\endgraf}\fi}
\newcommand{\rats}{\mbox{$\mathbb Q$}} 
\newcommand{\nats}{\mbox{$\mathbb N$}}
\newcommand{\ints}{\mbox{$\mathbb Z$}}
\newlength{\tablength}
\newlength{\spacelength}
\newcommand{\tabstar}{\hspace*{\tablength}}
\newcommand{\spacestar}{\hspace*{\spacelength}}
\newcommand\vol{-\!\!\!\! V}  % Ray hacked in the volume symbol
\def\obeytabs{\catcode`\^^I=\active}
{\obeytabs\global\let^^I=\tabstar}
{\obeyspaces\global\let =\spacestar}
\newenvironment{display}{\begingroup\obeylines\obeyspaces\obeytabs}{\endgroup}
\newenvironment{prog}{\begin{display}\parskip0pt\sf}{\end{display}}

\title{The complexity of cyber attacks in a new layered-security model\\
and the maximum-weight, rooted-subtree problem}

\author{
{\sl Geir Agnarsson}\\
Department of Mathematical Sciences\\
George Mason University\\
Fairfax, VA 22030\\
{\tt geir@math.gmu.edu}
\and
{\sl Raymond Greenlaw}\\ Cyber Security Studies\\
United States Naval Academy\\
Annapolis, Maryland 21402\\
{\tt greenlaw@usna.edu}
\and
{\sl Sanpawat Kantabutra}\\
Computer Engineering Department\\
Chiang Mai University\\
Chiang Mai, 50200, Thailand\\
{\tt sanpawat@alumni.tufts.edu}
}

\date{\today}

\begin{document}

\maketitle

\begin{abstract}
This paper makes three contributions to cyber-security research.
First, we define a model for cyber-security systems and
the concept of a {\em cyber-security attack\/} within the model's
framework.  The model highlights the importance of {\em game-over
  components\/}---critical system components which if acquired will
give an adversary the ability to defeat a system completely.  The
model is based on systems that use defense-in-depth/layered-security
approaches, as many systems do.  In the model we define the concept of
{\em penetration cost}, which is the cost that must be paid in order
to break into the next layer of security.  Second, we define natural
decision and optimization problems based on cyber-security attacks in
terms of doubly weighted trees, and analyze their complexity. More
precisely, given a tree $T$ rooted at a vertex $r$, a {\em penetrating
  cost} edge function $c$ on $T$, a {\em target-acquisition} vertex
function $p$ on $T$, the attacker's {\em budget} and the {\em
  game-over threshold} $B,G \in {\rats}^{+}$ respectively, we
consider the problem of determining the existence of a rooted subtree
$T'$ of $T$ within the attacker's budget (that is, the sum of the
costs of the edges in $T'$ is less than or equal to $B$) with total
acquisition value more than the game-over threshold (that is, the sum
of the target values of the nodes in $T'$ is greater than or equal to
$G$).  We prove that the general version of this problem is
intractable, but does admit a polynomial time approximation scheme.
We also analyze the complexity of three restricted
versions of the problems, where the penetration cost is the constant
function, integer-valued, and rational-valued among a given fixed
number of distinct values. Using recursion and dynamic-programming
techniques, we show that for constant penetration costs an {\em
  optimal\/} cyber-attack strategy can be found in polynomial time,
and for integer-valued and rational-valued penetration costs {\em
  optimal\/} cyber-attack strategies can be found in pseudo-polynomial
time.  Third, we provide a list of open problems relating to the
architectural design of cyber-security systems and to the model.
\end{abstract}

{\bf Keywords:} 
cyber security, defense-in-depth,
game over, 
information security, 
layered security, 
weighted rooted trees, 
complexity, 
polynomial time,
pseudo-polynomial time.
% Ray deleted these two NP-hardness, NP-completeness, as these subjects 
% are no longer and vogue, and our list is now quite long.

\section{Introduction}
\label{sec:introduction}

Our daily life, economic vitality, and a nation's security depend on a
stable, safe, and secure cyberspace.  Cyber security is so important
that the United States (US) Department of Defense established the
US~Cyber Command to take charge of pulling together existing
cyberspace resources, creating synergy, and synchronizing war-fighting
efforts to defend the information-security environment of the
US~\cite{jackson}.  Other countries also have seen the importance of cyber
security.  To name just a few in what follows, in response to North
Korea's creation of a cyber-warfare unit, South Korea created a
cyber-warfare command in December~2009~\cite{koreacyber}.  During
2010, China introduced its first department dedicated to defensive
cyber war and information security in response to the creation of the
US~Cyber Command~\cite{branigan}.  The United Kingdom has also stood
up a cyber force~\cite{brown}.  Other countries are quickly following
suit.

Cyberspace has become a new frontier that comes with new
opportunities, as well as new risks.  According to a 2012 study of
US~companies, the occurrence of cyber attacks has more than doubled
over a 3-year period while the adverse financial impact has increased
by nearly 40 percent~\cite{ponemon}.  More specifically,
US~organizations experienced an average of 50, 72, and 102 successful
attacks against them per week in 2010, 2011, and 2012, respectively.
In~\cite{sparrows} a wide range of cyber-crime statistics are
reported, including locations of attacks, motivation behind attacks,
and types of attacks.  The number of cyber attacks is increasing
rapidly, and for the month of June 2013, 4\% of attacks were
classified as cyber warfare, 8\% as cyber espionage, 26\% as
hacktivism, and 62\% as cyber crime (see~\cite{sparrows}).  Over the
past couple of years these percentages have varied significantly from
month-to-month.  In order to respond to cyber attacks, organizations
have spent increasing amounts of time, money, and energy at levels
that are now becoming unsustainable.  Despite the amounts of time,
money, and energy pouring into cyber security, the field is still
emerging and widely applicable solutions to the problems in the field
have not yet been developed.

A secure system must defend against all possible cyber attacks,
including zero-day attacks that have never been known to the
defenders.  But, due to limited resources, defenders generally develop
defense systems for the attacks that they do know about.  Their
systems are secure to known attacks, but then become insecure as new
kinds of attacks emerge, as they do frequently.  To build a secure
system, therefore, requires first principles of security.  ``In other
words, we need a {\em science of cyber security\/} that puts the
construction of secure systems onto a firm foundation by giving
developers a body of laws for predicting the consequences of design
and implementation choices''~\cite{schneider}.  To this end Schneider
called for more models and abstractions to study cyber
security~\cite{schneider}.  In his article Schneider suggested
building a science of cyber security from existing areas of computer
science.  In particular, he mentioned formal methods, fault-tolerance,
cryptography, information theory, game theory, and experimental
computer science.  All of these subfields of computer science are
likely to be valuable sources of abstractions and laws.

Cyber security presents many new challenges.  Dunlavy et al.~discussed
what they saw as some of the major mathematical problems in cyber
security~\cite{dunlavy:hendrickson:kolda}.  One of the main
challenges is modeling large-scale networks using explanatory and
predictive models.  Naturally, graph models were proposed.  Some
common measures of a graph that such a model would seek to emulate are
distribution over the entire graph of vertex in-degrees and out-degrees,
graph diameter, community structure, and evolution of any of the
mentioned measures over time~\cite{chakrabarti:faloutsos}.  Pfleeger
discussed a number of useful cyber-security metrics~\cite{pfleeger}.
She introduced an approach to cyber-security measurement that uses a
multiple-metrics graph as an organizing structure by depicting the
attributes that contribute to overall security, and uses a process
query system to test hypotheses about each of the goals based on
metrics and underlying models.  Rue, Pfleeger, and Ortiz developed a
model-evaluation framework that involves making explicit each model's
assumptions, required inputs, and applicability
conditions~\cite{rue:pfleeger:ortiz}.

Complexity science, which draws on biological and other natural
analogues, seems under utilized, but perhaps is one of the
more-promising approaches to understanding problems in the
cyber-security domain~\cite{armstrong:mayo:siebenlist}.  Armstrong,
Mayo, and Siebenlist suggested that models of complex cyber systems
and their emergent behavior are needed to solve the problems arising
in cyber security~\cite{armstrong:mayo:siebenlist}.  Additionally,
theories and algorithms that use complexity analysis to reduce an
attacker's likelihood of success are also needed.  Existing work in
the fields of fault tolerance and high-reliability systems are
applicable too.  Shiva, Roy, and Dasgupta proposed a cyber-security
model based on game theory~\cite{shiva:roy:dasgupta}.  They discovered
that their model works well for a dynamically-changing scenario, which
often occurs in cyber systems.  Those authors considered the interaction
between the attacks and the defense mechanisms as a game played
between the attacker and the defender.

This paper is our response to the call for more cyber-security models
in~\cite{schneider}.  This work also draws attention to the importance
of designing systems that do not have {\em game-over
  components}---components that are so important that once an
adversary has taken them over, one's system is doomed. Since, as we
will see, such systems can be theoretically hacked fairly efficiently.
We model (many known) security systems mathematically and then discuss
their vulnerabilities.  Our model's focus is on systems having layered
security; each security layer possesses valuable assets that are kept
in {\em containers\/} at different levels.  An attacker attempts to
break into these layers to obtain assets, paying penetration costs
along the way in order to break in, and wins if a given game-over
threshold is surpassed before the attacker's budget runs out.  A given
layer of security might be, for example, a firewall or encryption.
The associated cost of by-passing the firewall or encryption is the
penetration cost that is used in the model.  We formalize the notion
of a cyber attack within the framework of the model.  For a number of
interesting cases we analyze the complexity of developing cyber-attack
strategies.

%% XXNB! Geir rewords slightly and fixes some snafus.. 17ap2015
The outline of this article is as follows.  In Section~\ref{sec:model}
we define the model for cyber-security systems, present an
equivalent weighted-tree view of the model, and define natural
problems related to the model.  A general decision problem (Game-Over
Attack Strategy, Decision Problem {\sc GOAS-DP}) based on the model is
proved NP-complete in Section~\ref{sec:general}; its corresponding
optimization problem ({\sc GOAS-OP}) is NP-hard.  In
sections~\ref{sec:constant-cost}, \ref{sec:int-cost-approx},
and~\ref{sec:rational-cost} we provide a polynomial-time algorithm for
solving {\sc GOAS-OP} when penetration costs are constant, a
pseudo-polynomial-time algorithm for solving {\sc GOAS-OP} when
penetration costs are integers, a polynomial-time approximation algorithm for
solving {\sc GOAS-OP} in general, and a polynomial-time algorithm 
for solving {\sc GOAS-OP} when penetration costs are
rational numbers from a prescribed finite collection of possible
rational costs, respectively. As an easy corollary,
we obtain a pseudo-polynomial-time algorithm for solving an optimization
problem on general weighted non-rooted trees.  
Table~\ref{tab:summary} summarizes the computational results 
of the paper.
\begin{table}[htb]
\begin{center}
\begin{tabular}{|l||c|c|}
\hline 
{\bf Problem Name} & {\bf Time} & {\bf Class} \\
\hline 
{\sc GOAS-DP} & -- & NP-complete\\
{\sc GOAS-OP} & -- & NP-hard\\
{\sc GOAS-DP} constant pc & $O(m^2n)$ & P\\
{\sc GOAS-OP} constant pc & $O(m^2n)$ & P\\
{\sc GOAS-DP} integer pc & $O(B^2 n)$ & pseudo-pt\\
{\sc GOAS-OP} integer pc & $O(B^2 n)$ & pseudo-pt\\
{\sc GOAS-OP} approx.    & $O((1/\epsilon)^2n^3)$ & P\\
{\sc GOAS-DP} rational pc & $O(m^{2d}n)$ & P\\
{\sc GOAS-OP} rational pc & $O(m^{2d}n)$ & P\\
\hline %% XXNB! Geir adds the approx results and makes the rational in P.! 17ap2015
\end{tabular}
\end{center}
\caption{Summary of results about the cyber-security model 
  contained in the paper.  Note that in the table ``pc'' stands for 
  ``penetration cost,'' and ``pseudo-pt'' stands for pseudo-polynomial time. 
  The values of $m$, $n$, $B$, and $d$ are as given in the respective 
  theorems.}
\label{tab:summary}
\end{table}
Conclusions and open problems are discussed in Section~\ref{sec:summary}.

\section{Model for Cyber-Security Systems}
\label{sec:model}

\subsection{Basic Setup}

When defining our cyber-security game-over model, we need
to strike a balance between simplicity and utility.  If the model is
too simple, it will not be useful to provide insight into real
situations; if the model is too complex, it will be cumbersome to
apply, and we may get bogged down in too many details to see the
forest from the trees.  In consultation with numerous cyber-security
experts, computer scientists, and others, we have come up with a good
compromise for our model between ease-of-use and the capability of
providing useful insights.

Many systems contain layered security or what is commonly referred to
as {\em defense-in-depth}, where valuable assets are hidden behind
many different layers or secured in numerous ways.  For example, a
{\em host-based defense\/} might layer security by using tools such as
signature-based vendor anti-virus software, host-based systems
security, host-based intrusion-prevention systems, host-based
firewalls, encryption, and restriction policies, whereas a {\em
  network-based defense\/} might provide defense-in-depth by using
items such as web proxies, intrusion-prevention systems, firewalls,
router-access control lists, encryption, and
filters~\cite{johnston:lafever}.  To break into such a system and
steal a valuable asset requires several levels of security to be
penetrated.  Our model focuses on this layered aspect of security and
is intended to capture the notion that there is a cost associated with
penetrating each additional level of a system and that attackers have
finite resources to utilize in a cyber attack.  We also build the
concept of critical game-over components.

\subsection{Definition of the Cyber-Security Game-Over Model}

Let $\nats = \{1,2,3,\ldots\}$, $\rats$ be the rational numbers, and
$\rats^+$ be the positive rational numbers. With the intuition provided 
in the previous section in mind, we now present the formal 
definition of the model.
\begin{definition}
\label{def:model}
A {\em cyber-security game-over model\/} $M$ is a
six-tuple ($\cal T$, $\cal C$, $\cal D$, $\cal L$, $B$, $G$), where 

\begin{enumerate}
\item
The set $\cal T$ $= \{ t_1,t_2,\ldots,t_k \}$ is a collection of 
{\em targets}, where $k \in \nats$\@.  The value $k$ is the 
{\em number of targets}.  Corresponding to each target $t_i$, 
for $1 \leq i \leq k$, is an associated 
{\em target acquisition value\/} $v(t_i)$, where
$v(t_i) \in \rats$\@.  We also refer to the target acquisition value
as the {\em acquisition value\/} for short, or as the {\em reward\/}
or {\em prize}.
\item
The set $\cal C$ $= \{ c_1, c_2, \ldots, c_l \}$ is a collection of
{\em containers}, where $l \in \nats$\@.  The value $l$ is the
{\em number of containers}.  Corresponding to each container $c_i$,
for $1 \leq i \leq l$, is an associated {\em penetration cost\/}
$p(c_i)$, where $p(c_i) \in \rats$\@.
\item
The set $\cal D$ $= \{ C_1, C_2, \ldots, C_l \}$ is the set of 
{\em container nestings}.  The tuple $C_i$, for $1 \leq i \leq l$, is
called the {\em penetration list\/} for container $c_i$ and is a list
in left-to-right order of containers that must be penetrated before
$c_i$ can be penetrated.  If a container $c_i$ has an empty
penetration list, and its cost $p(c_i)$ {\em has been paid}, we say that the
{\em container has been penetrated}.  
If a container $c_i$ has a
non-empty penetration list and each container in its list has been
{\em penetrated\/} in left-to-right order, and its cost $p(c_i)$ has
been paid, we say that the {\em container has been penetrated}.  The
number of items in the tuple $C_i$ is referred to as the 
{\em depth of penetration required for\/} $C_i$.  
%The {\em maximum required penetration depth for a model\/} is the 
%maximum penetration depth required for any container. XXNB! never
%used (ref #2) Geir 14ap2015
If container 
$c_j$ appears in $c_i\mbox{\rm 's}$ tuple $C_i$, we say that container 
$c_i$ {\em is dependent on container $c_j$}.  If there are no two 
containers $c_i$ and $c_j$ such that container $c_i$ is dependent 
on container $c_j$ and container $c_j$ is dependent on container 
$c_i$, then we say the {\em model is well-formed}.
\item 
The set $\cal L$ $= \{ l_1, l_2, \ldots, l_k \}$ is a list of
container names.  These containers specify the {\em level-1
  locations\/} of the targets.  For $1 \leq i \leq k$ if target $t_i$
has level-1 location $l_i$, this means that there is no other
container $\widehat{c}$ such that container $\widehat{c}$ is dependent
on container $l_i$ and container $\widehat{c}$ contains target $t_i$.
Target $t_i$ is said to be {\em located at level-1 in container\/}
$l_i$.  The target $t_i$ is also said to be {\em located in
  container\/} $l_i$ or any container on which container $l_i$ is
dependent. When a target's level-1 container has been penetrated, we
say that the {\em target has been acquired}.
\item
The value $B \in \rats$ is the {\em attacker's budget}.  The value
represents the amount of resources that an attacker can spend on a
cyber attack.
\item
The value $G \in \rats$ is the {\em game-over threshold\/} 
signifying when critical components have been acquired.  
\end{enumerate}
\end{definition}
The focus of this paper is on cyber-security game-over
models that are well-formed, which are motivated by real-world
scenarios.  In the next section we introduce a graph-theoretic version
of the model using weighted trees.

{\sc Remarks:} (i) In part~3 of the definition we refer to the cost of
a container $c_i$ being paid.  By this we simply mean that
$p(c_i)$ has been deducted from the remaining budget, $B'$, and we
require that $B' - p(c_i) \geq 0$.  (ii) In part~4
of the definition we maintain a general notion of containment for
targets by specifying the inner-most container in which a target is
located.  Although containers can have partial overlap, we require
that the inner-most container be unique.
In the next definition we formalize the notion of a 
{\em cyber-security attack strategy}.
\begin{definition}
\label{def:attack}
A {\em cyber-security attack strategy\/} in a cyber-security game-over
model $M$ is a list of containers $c_1, c_2, \ldots, c_r$
from $M$\@.  The {\em cost of an attack strategy\/} is 
$\sum_{i=1}^{r}p(c_i)$.  A {\em valid attack strategy\/} is one in which the
penetration order is not violated.  A {\em game-over attack strategy
  in a cyber-security game-over model\/} $M$ is a valid
attack strategy $c_1, c_2, \ldots, c_r$ whose cost is less than or
equal to $B$ and whose total target acquisition value 
$\sum_{i=1}^{r}v(t_i) \geq G$.  We call such a game-over attack strategy
in a cyber-security game-over model a {\em (successful)
  cyber-security attack\/} or {\em cyber attack\/} for short.
\end{definition}
Note that this notion of a cyber attack is more general than some,
and, for example, espionage would qualify as a cyber attack under this
definition.  The definition does not require that a service or network
be destroyed or disrupted.  Since many researchers will think of
Definition~\ref{def:model} from a graph-theory point of view, in the
next section we offer that perspective.  As we will soon see, the
graph-theoretic perspective allows us to work more easily with the
model mathematically and to relate to other known results.

\subsection{Game-Over Model in Terms of Weighted Trees}
\label{sec:weighted-trees}

In this section we describe the (well-formed) game-over model in terms
of weighted trees.  The set ${\cal D}$ of nested containers in
Definition~\ref{def:model} has a natural rooted-tree structure, where
each container corresponds to a vertex that is not the root, and we
have an edge from a parent $u$ down to a child $v$ if and only if the
corresponding container $c(u)$ includes the container $c(v)$ in it.
The weight of an edge from a parent to a child represents the cost of
penetrating the corresponding container.  The weight of a vertex
represents the acquisition value/prize/reward obtained by
penetrating/breaking into that container.

Sometimes we do not distinguish a target from its acquisition
value/prize/reward nor a container from its penetration cost.  We can
assume that the number of containers and targets is the same.  Since
if we have a container housing another container (and nothing else),
we can just look at this ``double'' container as a single container of
penetration cost equal to the sum of the two nested ones.  Also, if a
container contains many prizes, we can just lump them all into a
single prize, which is the sum of them all.  The following
is a graph-theoretic version of Definition~\ref{def:model}.
\begin{definition}
\label{def:GO-tree}
A {\em cyber-security (game-over) model (CSM)\/} $M$ is
given by an ordered five tuple $M = (T, c, p, B, G)$, where $T$ is a
tree rooted at $r$ having $n \in \nats$ non-root vertices, 
$c : E(T) \rightarrow {\rats}$ is a penetration-cost weight function, 
$p : V(T) \rightarrow {\rats}$ is the target-acquisition-value weight 
function, and $B, G \in {\rats}^{+}$ are the attacker's budget and the
game-over threshold value, respectively.
\end{definition}
{\sc Remarks:} (i) Note that $V(T) = \{r, u_1,\ldots, u_n\}$, where $r$
is the designated root that indicates the start of an attack.  (ii)
In most situations we have the weights $c$ and $p$ being non-negative
rational numbers, and $p(r) = 0$.

Recall that in a rooted tree $T$ each non-root vertex $u \in V(T)$ has
exactly one parent.  We let $e(u) \in E(T)$ denote the unique edge
connecting $u$ to its parent.  For the root $r$, we let $e(r)$ be the
empty set and $c(e(r))$ be $0$.  For a tree $T$ with $u \in V(T)$, we
let $T(u)$ denote the (largest) subtree of $T$ rooted at $u$.  It is
easy to see the correspondence between Definitions~\ref{def:model}
and~\ref{def:GO-tree}.  Analogously to Definition~\ref{def:attack}, we
next define a {\em cyber-security attack strategy\/} in the
weighted-tree model.
\begin{definition}
\label{def:GO-attack}
A {\em cyber-security attack strategy (CSAS)\/} in a CSM 
$M = (T, c, p, B, G)$ is given by a subtree $T'$ of $T$ 
that contains the root $r$ of $T$.
\begin{itemize}
\item
We define the {\em cost\/} of a CSAS $T'$ to be 
$c(T') = \sum_{u\in V(T')}c(e(u))$.
\item
We define a {\em valid CSAS (VCSAS)\/} to be a CSAS $T'$ with $c(T') \leq B$.
\item
We define the {\em prize\/} of a CSAS $T'$ to be 
$p(T') = \sum_{u\in V(T')}p(u)$.
\end{itemize}
A {\em game-over attack strategy (GOAS)\/} in a CSM $M = (T, c, p, B,
G)$ is a VCSAS $T'$ with $p(T') \geq G$\@.  We sometimes refer to such
a GOAS simply as a {\em cyber-security attack\/} or {\em cyber
  attack\/} for short.
\end{definition}
Note that in Definition~\ref{def:GO-attack} we use $c$ (resp.~$p$) to denote
the total cost (respectively, total prize) of a cyber-security attack
strategy.  We also use $c$ (resp.~$p$) as the penetration-cost weight
function (respectively, target-acquisition-value weight function).  The
overloading of this notation should not cause any confusion.  Throughout
the remainder of the paper, we will use Definitions~\ref{def:GO-tree}
and~\ref{def:GO-attack}.

\subsection{Cyber-Attack Problems in the Game-Over Model}
\label{sub:decision-problems}

We now state some natural questions based on the CSM.
\begin{problem}
\label{def:csam}
{\sc Given:} A cyber-security model $M = (T, c, p, B, G)$.
\begin{itemize}
\item
{\sc Game-Over Attack Strategy, Decision Problem} {\sc (GOAS-DP):}\\
Is there a game-over attack strategy in $M$? 
\item
{\sc Game-Over Attack Strategy, Optimization Problem} {\sc (GOAS-OP):}\\
What is the maximum prize of a valid game-over attack strategy  in $M$?
\end{itemize}
\end{problem}
Needless to say, some special cases are also of interest, in
particular, in Problems~\ref{def:csam} when $c$ is (i) a
constant rational function, (ii) an integer-valued function, or (iii)
takes only finitely many given rational values. We explore
the general GOAS and these other questions in the following sections.

\subsection{Some Limitations of the Model}
\label{sub:limitations}

Our model is a theoretical model.  It is designed to give us a deeper
understanding of cyber attacks and cyber-attack strategies.  Of
course, a real adversary is not in possession of complete knowledge
about a system and its penetration costs.  Nevertheless, it is
interesting to suppose that an adversary is in possession of all of
this information, and then to see what an adversary is capable of
achieving under these circumstances.  Certainly an adversary with less
information could do no better than our fully informed adversary.

We are considering systems as they are.  That is, we are given some
system, targets, and penetration costs.  If the system is a real
system, we are not concerned about how to improve the security of that
system per se.  We assume that the system is already in a hardened
state.  We then examine how difficult it would be to attack such a
system.  We do not examine the question of implementations of a
system.  Our model can be used on any existing system.  Some real
systems will have more than one possible path to attack a target.
And, in the future it may be worth generalizing the model to
structures other than trees.  The first step is to look at trees and
derive some insight from these cases.

We have purposely chosen a target acquisition function which is
simple.  That is, we merely add together the total costs of the
targets acquired.  Studying this simple acquisition function is the
first step.  It may be interesting to study more-complex acquisition
functions in the future.  For example, one can imagine two targets
that in and of themselves are of no real value, but when the
information contained in the two are combined they are of great value.
In some cases our additive function can capture this type of target
depending on the structure of the model.

We describe the notion of a game-over component.  In the model this
concept is an abstract one.  A set of components whose total value
exceeds a given threshold comprise a ``game-over component.''  A
game-over component is not necessarily a single target although one
can think of a high-cost target, which is included as a target in a
set of targets that push us over the game-over threshold, as being the
game-over component.

%% XXNB! Geir adds, by the suggestion of one of the refs 14ap2015
For easy reference, the following table contains our most common 
abbreviations, their spelled out meaning, and where they are defined.
\begin{table}[htb]
\begin{center}
\begin{tabular}{|l|c|c|}
\hline 
CSM  & cyber-security (game-over) model & Def.~\ref{def:GO-tree} \\
\hline 
CSAS & cyber-security attack strategy & Def.~\ref{def:GO-attack} \\
\hline 
VCSAS & valid cyber-security attack strategy & Def.~\ref{def:GO-attack} \\
\hline 
GOAS  & game-over attack strategy & Def.~\ref{def:GO-attack} \\
\hline
{\sc GOAS-DP} & game-over attack strategy, decision problem & Def.~\ref{def:csam} \\
\hline
{\sc GOAS-OP} & game-over attack strategy, optimization problem & Def.~\ref{def:csam} \\
\hline
\end{tabular}
\end{center}
\caption{Abbreviations we use throughout the paper, all
defined in this section.}
\label{tab:abbrev}
\end{table}

\section{Complexity of Cyber-Attack Problems}
\label{sec:general}

In this section we show that the general game-over attack strategy
problems are intractable, that is, highly unlikely to be amenable to
polynomial-time solutions.  Consider a cyber-security attack model
$M$, where $T$ is a star centered at $r$ having $n$ leaves
$u_1,\ldots,u_n$.  Since each cyber-security attack $T'$ of $M$ can be
presented as a collection $E' \subseteq E(T)$ of edges of $T$, and
hence also as a collection of vertices $V' \subseteq V(T)$ by $T' =
T[\{r\}\cup V']$, and vice versa, each collection of vertices
$V'\subseteq V(T)$ can be presented as $V' = V(T')$ for some
cyber-security attack $T'$ of $M$, and the {\sc GOAS-DP} is exactly the
decision problem of the $0/1$-{\sc Knapsack Problem}~\cite{garey:johnson}, 
and the {\sc GOAS-OP} is the
optimization problem of the {\sc Knapsack Problem}.  Note that the
$0/1$-{\sc Knapsack Problem} is usually stated using natural numbers
as weights, but clearly the case for weights consisting of rational
numbers is no easier to solve yet still in NP\@.  So, we have the
following observation.
\begin{observation}
\label{obs:GOSC=NP}
The {\sc GOAS-DP} is NP-complete;
the {\sc GOAS-OP} is an NP-hard optimization problem. 
\end{observation}
{\sc Remark:} Observation~\ref{obs:GOSC=NP} answers an open 
question in the last section of~\cite{DO-main}, where it is asked 
whether or not the {\sc LST-Tree Problem} can be solved in polynomial 
time (we presume) for general edge 
lengths. Observation~\ref{obs:GOSC=NP} is similar 
to~\cite[Theorem 2]{coene:trees}, where also a star is considered to 
show that their {\sc SubtreeE}  is as hard as {\sc Knapsack}.

Notice that the NP-completeness of {\sc GOAS-DP} is a double-edge
sword. It suggests that even an attacker who has detailed knowledge of
the defenses of a cyber-security system would find the problem of
allocating his (attack) resources difficult. On the other hand, the
NP-completeness also makes it difficult for the defender to assess the
security of his system. However, we will see in 
Section~\ref{sec:int-cost-approx},
that if we allow %% XXNB! Geir adds this to the remark. 13ap2015
a slight proportional increase of the attacker's budget $B$ to an amount of 
$(1+\epsilon)B$ for an $\epsilon\geq 0$, then {\sc GOAS-OP} 
admits a polynomial time approximation scheme, so it can be solved in 
time polynomial in $n$ and $1/\epsilon$.

Sections~\ref{sec:constant-cost}, \ref{sec:int-cost-approx},
and~\ref{sec:rational-cost} consider the complexity of
cyber-security attacks where $c$ is a constant-valued cost function,
an integer-valued cost function, and a rational-valued cost function
of finitely many possible values, respectively. In 
Section~\ref{sec:int-cost-approx}, as mentioned, we also obtain an 
approximation 
algorithm for solving {\sc GOAS-OP}, and a solution on general weighted
non-rooted trees. In all cases we are
able to give reasonably efficient algorithms for solving 
{\sc GOAS-OP}. 

\section{Cyber Attacks with Constant Penetration Costs}
\label{sec:constant-cost}

In this section we show that if all penetration costs have the same
value then the {\sc Game-Over Attack Strategy Problems} can be solved
efficiently in polynomial time.  Consider a CSM $M$, where $c$ is a
constant function taking a constant rational value $c(e) = c$ for each
$e\in E(T)$.  That is, all penetration costs are a fixed-rational
value.  This variant is the first interesting case of the {\sc
  GOAS-DP} and {\sc GOAS-OP}, as there are related problems and
solutions in the literature.  One of the first papers on
maximum-weight subtrees of a given tree with a specific root
is~\cite{Houssaine-etal}, where it is shown that the {\em rooted
  subtree problem}, that is, to find a maximum-weight subtree with a
specific root from a given set of subtrees, is in polynomial time if,
and only if, the {\em subtree packing problem}, that is, to find
maximum-weight packing of vertex-disjoint subtrees from a given set of
subtrees (where the value of each subtree can depend on the root), is
in polynomial time. In more-recent papers the {\em weight-constrained
  maximum-density subtree problem (WMSP)} is considered: given a tree
$T$ having $n$ vertices, and two functions $l, w : E(T)\rightarrow
{\rats}$ representing the ``length'' and ``weight'' of the edges,
respectively, determine the subtree $T'$ of $T$ such that $\sum_{e\in
  E(T')}w(e)/\sum_{e\in E(T')}l(e)$ is a maximum, subject to
$\sum_{e\in E(T')}w(e)$ having a given upper bound.
In~\cite{LNCS-bounded} an $O(w_{\max}n)$-time algorithm is given to
solve the related, and more restricted, {\em weight-constrained
  maximum-density path problem (WMPP)}, as well as an
$O(w_{\max}^2n)$-time algorithm to solve the WMSP\@.
In~\cite{DO-main} an $O(nU^2)$-time algorithm is given for the WMSP,
where $U$ is the maximum total length of the subtree, and
in~\cite{IPL-improved} an $O(nU\lg n)$-time algorithm for the WMSP is
given, which is an improvement in the case when $U = \Omega(\lg n)$.
The WMSP has a wide range of practical applications.  In particular, the
related WMPP has applications in computational
biology~\cite{LNCS-bounded}, and the related {\em weight-constrained
  least-density path problem (WLPP)} also has applications in
computational biology, as well as in computer, traffic, and logistic
network designs~\cite{DO-main}.

The WMSP is similar to our problem, and some of the same approaches
used in~\cite{LNCS-bounded}, \cite{DO-main}, and~\cite{IPL-improved}
can be applied in our case, namely the techniques of recursion and
dynamic programming.  There are not existing results that apply
directly to our problems.  Note that there is a subtle difference
between our {\sc GOAS-OP} and the {\sc WMSP}, as a maximum-weight
subtree (that is, with the prize $p(T')$ a maximum) might have low
density and vice versa; a subtree of high density might be ``small''
with low total weight (that is, prize).  

In~\cite{coene:trees} a problem on trees related to the 
Traveling Salesman Problem
with profits is studied, which is similar to what we do. 
\comment{ %XX To much detail.! Geir.
To describe their work, 
and relate to ours, define a partial order $\preceq$ on ${\rats}^2$
by $(x,y)\preceq (x',y') \Leftrightarrow x\geq x' \mbox{ and } y\leq y'$.
In~\cite{coene:trees} their main objective is to obtain the maximal
points (the {\em efficient points} as they call them) w.r.t.~$\preceq$, 
whereas we
are interested in finding the points $\succeq (B,G)$ for given $B$ and $G$. 
Hence, although the dynamic algorithm approach both here 
and in~\cite{coene:trees} is similar, this yields to a procedural
difference of the actual algorithms. 
}
Both here and in~\cite{coene:trees} 
the most general form of the problems considered, in our case {\sc GOAS-DP}
in Observation~\ref{obs:GOSC=NP} and in their case (as mentioned 
above) {\sc SubtreeE} in~\cite[Theorem 2]{coene:trees}, 
are observed to be as hard as {\sc Knapsack} and hence NP-complete. 
Also, the results of fixed costs, in our case
Theorem~\ref{thm:c=const} and in their case~\cite[Theorem 3]{coene:trees}, 
the problems are shown to be solvable in $O(n)$ time, given certain conditions.
Theorem~\ref{thm:c=const}, however, provides a precise accounting for the time
complexity and for certain values of $m$, defined there, our algorithm
would be faster than that given in~\cite{coene:trees}.  Their work is
not in the context of cyber-security, and does not handle cases as
general as this work.

For a CSM $M$, where $c$ is a constant function, we first note that
$T'$ is a VCSAS if and only if $m = |E(T')| \leq \lfloor B/c\rfloor$.
Hence, in this case the {\sc GOAS-OP} reduces to finding a CSAS $T'$
with at most $m$ edges having $p(T')$ at a maximum.  Note that if
$m\geq n$, then the {\sc GOAS-OP} is trivial since $T' = T$ is the
optimal subtree.  Hence, we will assume the budget $B$ is such that 
$m < n$.

%% XXNB! Geir added this summary/outline of what we do. 14ap2015
In what follows, we will describe our dynamic programming setup 
to solve {\sc GOAS-OP} in this case. The core of the idea is simple:
we construct a $2\times 2$ matrix for each vertex $u$ in the tree $T$ 
that stores the maximum prize of a subtree rooted at $u$ on at most
$k$ edges and that contains only the rightmost $d(u) - i + 1$ 
branches from $u$, for each $k\in\{1,\ldots,m\}$ and  $i\in\{1,\ldots,d(u)\}$.

More specifically, we proceed as follows. We may 
assume that our rooted tree $T$ has its vertices ordered from
left-to-right in some arbitrary but fixed order, that is, $T$ is a
{\em planted plane tree}.  Since $T$ has $n\geq 1$ non-root vertices
and $n+1$ vertices total, we know by a classic counting
exercise~\cite{geirray} that the number of planted plane trees on
$n+1$ vertices is given by the Catalan numbers $C_n$
% $C_n = \frac{1}{n+1}\binom{2n}{n}$, 
by obtaining a defining recursion for
$C_n$ by decomposing each planted plane tree into two rooted
subtrees. Using this decomposition, we introduce some notation.  For a
subtree $\tau$ of $T$ rooted at $u\in V(T)$ denote by $\tau(v)$ the
largest subtree of $\tau$ that is rooted at a vertex $v$ (if $v \in
T[V( \tau)]$).  Denote by $u_{\ell}$ the leftmost child of $u$ in
$\tau$ (if it exists).  Let $\tau_{\ell} = \tau(u_{\ell})$ denote the
subtree of $\tau$ generated by $u_{\ell}$, that is, the largest
subtree of $T$ rooted at $u_{\ell}$.  Finally, let $\tau'' = \tau -
V(\tau_{\ell}) = T[V(\tau)\setminus V(\tau_{\ell})]$ denote the
subtree of $\tau$ generated by the vertices not in $\tau_{\ell}$.  In
this way we obtain a decomposition/partition of the planted plane tree
$\tau$ into two vertex-disjoint subtrees $\tau_{\ell}$ and $\tau''$
whose roots are connected by a single edge $e(u_{\ell})$.  In
particular, for each vertex $u\in V(T)$, we have a partition of $T(u)$
into $T(u)_{\ell} = T(u_{\ell})$ and $T(u)''$, which we will denote by
$T''(u)$ (that is $T(u)'' = T''(u)$).  Note that if $u$ is a leaf,
then $T(u) = T''(u) = \{u\}$ and $u_{\ell} = T(u_{\ell}) =
\emptyset$. Also, if $u$ has exactly one child, which therefore is its
leftmost child $u_{\ell}$, then $T(u)$ is the two-path between $u$ and
its only child $u_{\ell}$, $T''(u) = \{u\}$, and $T(u_{\ell}) =
\{u_{\ell}\}$. Assuming the degree of $u$ is $d(u)$, we can
recursively define the trees $T^1(u),\ldots, T^{d(u)}(u)$ by
\begin{eqnarray*}
T^1(u) & = & T(u), \\
T^{i+1}(u) & = & (T^i)''(u).
\end{eqnarray*}
For each vertex $u\in V(T)$, we create a $d(u)\times (m+1)$ rational matrix 
as follows:
\[
\mathbf{M}(u)= 
\left[
\begin{array}{cccc}
  M_0^1(u) & M_1^1(u) & \cdots & M_m^1(u) \\
  M_0^2(u) & M_1^2(u) & \cdots & M_m^2(u) \\
           &         & \vdots  &          \\
  M_0^{d(u)}(u) & M_1^{d(u)}(u) & \cdots & M_m^{d(u)}(u) 
\end{array}
\right],
\]
where $M_k^i(u)$ is the maximum prize of a subtree of $T^i(u)$ rooted
at $u$ with at most $k$ edges for each $i\in \{1,\ldots,d(u)\}$ and
$k\in\{0,1,\ldots, m\}$.  In particular, $M_0^i(u) = p(u)$ for each
vertex $u$ and $i\in\{1,\ldots,d(u)\}$.  For each leaf $u$ of $T$, and
each $i$ and $k$, we set $M_k^i(u) = p(u)$, and for each internal
vertex $u$ we have a recursion given in the following way: for a
vertex $u$ and an {\em arbitrary\/} subtree $\tau$ rooted at $u$, we
let $M_k(u;\tau)$ be the maximum prize of a subtree of $\tau$ rooted
at $u$ having $k$ edges or $0$ if vertex $u$ does not exist.  If a
maximum-prize subtree of $\tau$ with $k$ edges does not contain the
edge from $u$ to its leftmost child $u_{\ell}$, then $M_k(u;\tau) =
M_k(u;\tau'')$.  Otherwise, such a maximum subtree contains $i-1$
edges from $\tau_{\ell}$ and $k-i$ edges from $\tau''$. The following
lemma is easy to show.
\begin{lemma}
\label{lmm:iff-const}
The arbitrary subtree $\tau$ rooted at $u$ is a maximum-prize subtree
with at most $k$ edges that contains the leftmost child $u_{\ell}$ of
$u$ if and only if the included subtree of $\tau_{\ell}$ is a
maximum-prize subtree with at most $i-1$ edges rooted at $u_{\ell}$
and the included subtree of $\tau''$ is a maximum-prize subtree with
at most $k-i$ edges rooted at $u$ for some $i\in\{1,\ldots, k\}$.
\end{lemma}
By Lemma~\ref{lmm:iff-const} we therefore have the following recursion:
\begin{equation}
\label{eqn:dyn-prog-rec}
M_k(u;\tau) = \max\left( M_k(u;\tau''), 
\max_{1\leq i\leq k} \left(M_{i-1}(u_{\ell};\tau_{\ell}) 
+ M_{k-i}(u;\tau'')\right)\right).
\end{equation}
Since now $M_k^i(u) = M_k(u;T^i(u))$ for each $i$ and $k$, we see that
we can compute each $M_k^i(u)$ from the smaller $M$'s as given in
(\ref{eqn:dyn-prog-rec}) using $O(k)$-arithmetic operations.  Because
$k\in\{0,1,\ldots,m\}$, this fact means in $O(m)$-arithmetic
operations.  Since we assume each arithmetic operation takes one step,
we have that each $M_k^i(u)$ can be computed in $O(m)$-time given the
required inputs.  Therefore, $\mathbf{M}(u)$ can be computed in
$d(u)m\cdot O(m) = d(u) O(m^2)$-time.  Performing these calculations for
each of the $n$ vertices of our given tree $T$, we obtain by the
Handshaking Lemma a total time of
\[
t(n) = \sum_{u\in V(T)} d(u) O(m^2) = O(m^2)\sum_{u\in V(T)} d(u) =
O(m^2)2(n-1) = O(m^2 n).
\]
We finally compute a maximum prize VCSAS $T'$ in $M$ by 
$p(T') = M_m^1(r)$ for the root $r$ of $T$.  We conclude by the 
following theorem.
\begin{theorem}
\label{thm:c=const}
If $M = (T, c, p, B, G)$ is a CSM, where $T$ has $n$ vertices, $c$ is
a constant function, and $m = \lfloor B/c\rfloor$ 
then the {\sc GOAS-OP}
can be solved in $O(m^2n)$-time. 
\end{theorem}
{\sc Remarks:} (i) 
Note that Theorem~\ref{thm:c=const} is similar to~\cite[Theorem 3]{coene:trees}.
(ii) Also note that the overhead constant is ``small'': for
each vertex $u$, each $k$, and each $i$ by (\ref{eqn:dyn-prog-rec})
each of $M_k^i(u) = M_k(u;T^i(u))$ uses exactly $2k$ arithmetic
operations, namely $k$ additions and $k$ comparisons.  Hence, the
exact number of arithmetic operations can, by the Handshaking Lemma,
be given by
\[
N(n,m) = \sum_{u\in V(T)}\sum_{k=0}^m d(u)(2k) 
= \sum_{u\in V(T)} d(u)\sum_{k=0}^m 2k = 2|E(T)|m^2 = 2(n-1)m^2.
\]
We obtain an overhead constant of two.  Since we assumed the
budget given is such that $m < n$, we see that the {\sc GOAS-OP}
can be solved in $O(n^3)$ time.
\comment{
We did agree of removing this portion of the remark? 
although I am sure it is valid, I personally don't want to get into
some argument/dialog with the authors of that paper... Geir

Yes, that is fine ... December 18, 2013.  Let us leave it out.  Ray

(ii) Just before Corollary 2 in~\cite{IPL-improved}, it is claimed
that the same approach used there can be modified to obtain an
$O(n^2)$-time algorithm to compute the maximum-weight subtrees of all
sizes.  In particular, in the paragraph right before Corollary 2 it
reads:
\begin{quote}
``In this way, the maximum-weight subtree of the size $i$ is the
  maximum of maximum-weight subtree of the size $i$ in each tree that
  has a size greater than or equal to $i$.''
\end{quote}
Such an argument might work when considering maximum-density, but not
when considering {\em maximum-weight} subtrees of a given weight, as
easy examples can show.\footnote{The rest of the paper appears correct
  and is well written.} Hence, to the best of our knowledge the
$O(n^2)$-time procedure presented in~\cite{IPL-improved} is faulty,
and so the best algorithm known to us for computing maximum-weight
subtrees of a given size is a modified version of the procedure
yielding Theorem~\ref{thm:c=const}.
}
\begin{corollary}
\label{cor:GOAS-DP}
The {\sc GOAS-DP} when restricted to constant-valued penetration costs
can be solved in $O(n^3)$ time and is in P.
\end{corollary}

\section{Cyber Attacks with Integer Penetration Costs 
and an Approximation Scheme}
\label{sec:int-cost-approx}

In this section we show that if all penetration costs are non-negative
integers then the {\sc Game-Over Attack Strategy Problems} can be
solved in pseudo-polynomial time. We will then use that to obtain
a polynomial time approximation algorithm. 

\subsection{Integer valued cost}

Consider now a CSM $M = (T, c, p, B, G)$, where $c$ is a non-negative 
integer-valued function, that is,
$c(e) \in \{0,1,2,\ldots\}$ for each $e\in E(T)$.  Note that we can
contract $T$ by each edge $e$ with $c(e) = 0$, thereby obtaining a
tree for our CSM $M$, where $c : E(T)\rightarrow \nats$ takes only
positive-integer values.  We derive a polynomial-time algorithm in
terms of $n$ and $B$ to solve the {\sc GOAS-OP}\@.  We can assume $B$
is an integer here as well since otherwise we could just replace
$B$ with $\lfloor B\rfloor$\@.  To produce our new algorithm we will
tweak the argument given in Section~\ref{sec:constant-cost} for the
case when the cost function $c$ is a constant.

Using the same decomposition of a subtree $\tau$ of $T$ into
$u_{\ell}$ and $\tau''$ for our dynamic programming scheme, 
for each vertex $u$ we will assign, as before, a
$d(u)\times (B+1)$ integer matrix as follows:
\[
\mathbf{N}(u)= 
\left[
\begin{array}{cccc}
  N_0^1(u) & N_1^1(u) & \cdots & N_B^1(u) \\
  N_0^2(u) & N_1^2(u) & \cdots & N_B^2(u) \\
           &          & \vdots &          \\
  N_0^{d(u)}(u) & N_1^{d(u)}(u) & \cdots & N_B^{d(u)}(u) 
\end{array}
\right],
\]
where $N_k^i(u)$ is the maximum prize of a subtree of $T^i(u)$ 
rooted at $u$  of total cost at most $k$ for each $i\in\{ 1,\ldots,d(u)\}$ 
and $k\in \{0,\ldots,B\}$. As before, we have $N_0^i(u) = p(u)$ for
each vertex $u$. Similarly to Lemma~\ref{lmm:iff-const}, we obtain
the following.
\begin{lemma}
\label{lmm:iff-nats}
The arbitrary subtree $\tau$ rooted at $u$ is a maximum-prize subtree
of total cost at most $k$ that contains the leftmost child $u_{\ell}$
of $u$ if and only if the included subtree of $\tau_{\ell}$ is a
maximum-prize subtree of total cost at most $i-c(e(u_{\ell}))$ rooted
at $u_{\ell}$ and the included subtree of $\tau''$ is a maximum-prize
subtree of total cost $k-i$ rooted at $u$, for some
$i\in\{c(e(u_{\ell})),\ldots,k\}$.
\end{lemma}
Using similar notation and definitions as in
Section~\ref{sec:constant-cost}, by Lemma~\ref{lmm:iff-nats} we get
the following recursion:
\begin{equation}
\label{eqn:dyn-prog-rec-C}
N_k(u;\tau) = \max\left( N_k(u;\tau''), 
\max_{c(e(u_{\ell}))\leq i\leq k} \left(N_{i-c(e(u_{\ell}))}(u_{\ell};\tau_{\ell}) 
+ N_{k-i}(u;\tau'')\right)\right),
\end{equation}
and we obtain similarly the following.
\begin{theorem}
\label{thm:c=nats}
If $M = (T, c, p, B, G)$ is a CSM, where $T$ has $n$ vertices and 
$c : E(T)\rightarrow {\nats}$ takes only positive-integer values, 
then the {\sc GOAS-OP} can be solved in $O(B^2n)$-time.
\end{theorem}
{\sc Remark:} (i) Although we are not able to obtain a compact
expression for the exact number of arithmetic operations that yield
Theorem~\ref{thm:c=nats}, the bound $N(n,B) = 2(n-1)B^2$ still is an
upper bound, as for Theorem~\ref{thm:c=const}.  (ii) Note the
assumption that $c$ is an {\em integer}-valued cost function is
crucial, since otherwise, we would not have been able to use the
recursion (\ref{eqn:dyn-prog-rec-C}) in at most $B$ steps.
\begin{corollary}
The {\sc GOAS-DP} when restricted to integer-valued penetration costs
can be solved in pseudo-polynomial time.
\end{corollary}

\subsection{Approximation Scheme} 
% XXNB! Geir adds this section 10ap2015 

We now can present a polynomial time approximation scheme (PTAS) for
solving the {\sc GOAS-OP} from Problem~\ref{def:csam}. In 
Observation~\ref{obs:GOSC=NP} we saw that the {\sc GOAS-OP} is 
an NP-hard optimization problem. But this is not the whole story;
although it is hard to compute the exact solution, one can obtain 
a polynomial time approximation algorithm if we allow slightly more budget
for the attacker than he/she wants to spend. We will in this
section describe one such approximation scheme. Our approach 
here is similar to the PTAS for the optimization of the 
$0/1$-{\sc Knapsack Problem} presented in the classic 
text~\cite[Section 17.3]{Papa-Steig}. 

We saw in Theorem~\ref{thm:c=nats} that {\sc GOAS-OP}
can be solved in $O(B^2n)$-time, if the cost is integer valued
and $B$ is the budget of the attacker. So for large $B$ this
can be far polynomial time. For each fixed $t\in\nats$
we can write the integer cost $c(e)$ of each edge $e\in E(T)$ as 
\begin{equation}
\label{eqn:cmod2^t}
c(e) = c_q(e) + c_r(e), \ \ \mbox{ where }c_r(e) = c(e) \bmod 2^t,
\end{equation}
that is, we obtain a new cost function $c_q$ by ignoring the
last $t$ digits of $c(e)$ when it is written as a binary number.
Since each $c_q$ is divisible by $2^t$, solving
{\sc GOAS-OP} for $c_q$ and budget $B$ is the same as solving it for
the cost function $2^{-t}c_q$ and budget $2^{-t}B$. Therefore, we can by 
Theorem~\ref{thm:c=nats} solve the {\sc GOAS-OP}
for this new cost function $c_q$ in $O((2^{-t}B)^2n)$-time.

Let $T'$ (resp.~$T_q'$) be an optimal {\sc GOAS-OP} subtree of $T$ 
w.r.t~the cost $c$ (resp.~$c_q$),  
so $p(T')$ is maximum among subtrees with $c$-weight $\leq B$,
and $p(T_q')$ is maximum among subtrees with $c_q$-weight $\leq B$.
In this case we have
\begin{equation}
\label{eqn:cost-q}
c(T_q') = c_q(T_q') + c_r(T_q') \leq B + |E(T_q')|\cdot 2^t 
\leq B + n2^t.
\end{equation}
Also, since $c_q(T') \leq c(T')\leq B$ we have by the definitions
of $T'$ and $T_q'$ that $p(T') \leq p(T_q')$. Therefore if there
is a GOAS $T'$ w.r.t.~the cost $c$, then there certainly 
is one w.r.t.~the cost $c_q$, namely $T_q$.
Hence, if $\epsilon = \frac{n2^t}{B}$, then we obtain from 
(\ref{eqn:cost-q}) that $c(T_q) \leq (1+\epsilon)B$ and
$T_q'$ is here definitely a GOAS that further can be computed
in $O((n/\epsilon)^2n) = O((1/\epsilon)^2n^3)$-time.
Conversely, for a given $\epsilon\geq 0$, we obtain such
an approximation algorithm by considering the cost $c_q$
defined by (\ref{eqn:cmod2^t}) where 
\begin{equation}
\label{eqn:t-digits}
t = \left\lfloor\lg\left(\frac{\epsilon B}{n}\right)\right\rfloor.
\end{equation}
We therefore have the following.
\begin{theorem}
\label{thm:GOAS-has-approx}
The {\sc GOAS-OP} admits a polynomial time approximation scheme;
for every $\epsilon \geq 0$ a GOAS $T'$ of cost of at most
$(1+\epsilon)B$ can be computed in $O((1/\epsilon)^2n^3)$-time.
\end{theorem}
{\sc Remarks:} (i) In establishing the above Theorem~\ref{thm:GOAS-has-approx}
we started with an integer cost function $c : E(T)\rightarrow {\nats}$.
The same approach could have been used for a rational
cost function $c : E(T)\rightarrow {\rats}$ where $c(e)$
has $d$ binary binary digits after its binary point 
(i.e.~radix point when written as a rational number in
base $2$.) By considering a new integer valued cost function
$c' : E(T)\rightarrow {\nats}$, where $c'(e) = 2^dc(e)$ for each 
$e\in E(T)$, we can
in the same manner as used above, obtain an approximation
algorithm where we replace $B$ with $B' = 2^dB$. Needless
to say however, in this case the corresponding cost function
$c_q'$ is obtained by truncating or ignoring only $t-d$ of the
digits of $c'$ (instead of the $t$ digits of $c$), 
to obtain a solution using a budged of $(1+\epsilon)B$.
(ii) Further along these lines, if the cost function 
$c : E(T)\rightarrow {\rats}$ is given as a fraction
$c(e) = a(e)/b(e)$, where $a(e), b(e)\in\nats$ are relatively prime,
we can let $M$ be the least common multiple of the $b(e)$ where $e\in E(T)$
and obtain by scaling by $M$ a new
integer valued cost function $c'' : E(T)\rightarrow {\nats}$ where
$c''(e) = Mc(e)$ for each $e\in E(T)$. Again, since $c''$ is integer valued
we can in the same manner  obtain an approximation 
algorithm where we replace $B$ with $B'' = MB$. In this
case the corresponding cost function $c_q''$ is obtained by
truncating or ignoring even fewer digits, namely $t - \lg M$ 
of the digits of $c''$. This will also yield a polynomial time approximation
algorithm in terms of $n$ and $1/\epsilon$ despite the fact that
$M$ can become very large (i.e.~if all the costs have pairwise
relatively prime denominators $b(e)$.)

\comment{ %% XXNB! not necessary.. Geir. 14ap2015
if $\epsilon$ is very small, that is
$\epsilon B/n$ is close to $1$, then, needless to say, $1/\epsilon$
is large and our approximation scheme is simply reduces to an exact solution
as in Theorem~\ref{thm:c=nats} as $c_q = c$, the solution of which, although not
polynomial in $n$ alone, is still polynomial in $n$ and $1/\epsilon$.
}

\subsection{General Weighted Trees}

In our framework a CSM $M$ is presented as a rooted tree provided
with two weight functions: one on the vertices and one on the edges.
In the model the root serves merely as a starting vertex and does not
(usually) carry any weight (that is, has no prize attached to it).
However, given a general non-rooted tree $T$ provided with two
edge-weight functions $w, w' : E(t) \rightarrow {\rats}$, we can
always add a root to some vertex and then push the weights of one of
the weight functions, say $w$ down to the unique vertex away from the
root.  In this way we obtain a CSM $M$ to which we can apply both
Theorems~\ref{thm:c=const} and~\ref{thm:c=nats}. With this slight
modification, we have the following corollary for general weighted
trees.
\begin{corollary}
\label{cor:general}
Let $T$ be a tree on $n$ vertices, $w, w' : E(T) \rightarrow {\rats}$
two edge-weight functions, and $B, G$ two rational numbers.  If the
function $w$ is either (i) a rational constant $c\in{\rats}$ 
or (ii) integer-valued, then the existence of a subtree $T'$ of $T$ 
such that $w'(T')\leq B$ and $w(T')$ is a maximum can be determined in 
$O(m^2n)$-time, where $m =\lfloor B/c\rfloor$ in case (i), 
and in $O(B^2n)$-time in case (ii).
\end{corollary}

\section{Cyber Attack with Rational Penetration Costs}
\label{sec:rational-cost}

In this section we consider the more-general case of a CSM 
$M = (T, c, p, B, G)$ where the cost function 
$c : E(T)\rightarrow {\rats}$ takes at most $d$ distinct 
rational values, say $c_1,\ldots,c_d\in\rats$.
This case can model quite realistic scenarios, as there are currently
only a finite number of known encryption methods and cyber-security
designs, where a successful hack for each method/design has a specific
penetration cost.  As in previous sections, we will utilize dynamic
programming and recursion based on the splitting of a subtree $\tau$
of a planted plane subtree into two subtrees $\tau_{\ell}$ and
$\tau''$ as in (\ref{eqn:dyn-prog-rec}) and
(\ref{eqn:dyn-prog-rec-C}).  However, here we are dealing with
rational-cost values (i.e.~arbitrary {\em real} values from all 
practical purposes), and that the we are able to obtain
a polynomial time procedure in this case is not as direct.

%% XXNB! Geir added this comment.
Note that if $M$ is the least common multiple of all the
denominators of $c_1,\ldots,c_d$, then by multiplying the
cost and the budget of the attacker through by $M$, we
obtain an integer valued cost function $Mc$, which then
can by Theorem~\ref{thm:c=nats} 
be solved pseudo polynomially in $O(M^2B^2n)$-time.
Our goal here in this section, however, is to develop an 
algorithm to solve {\sc GOAS-OP} in time polynomial
in $n$ alone.

For each $i\in \{1,\ldots,d\}$, let $n_i = |\{e\in E(T) : c(e) =
c_i\}|$, and so $\sum_{i=1}^d n_i = n = |E(T)| = |V(T)| - 1$.  Let
$\mathcal{B} = \{0,1,\ldots,n_1\}\times\cdots\times\{0,1,\ldots, n_d\}
\subseteq {\ints}^d$, and note that $|\mathcal{B}| = \prod_{i=1}^d(n_i
+ 1)$.  Denote a general $d$-tuple of ${\rats}^d$ by $\tilde{x} =
(x_1,\ldots,x_d)$, and let $\tilde{x}\leq\tilde{y}$ denote the usual
component-wise partial order $x_i\leq y_i$, for each
$i\in\{1,\ldots,d\}$. If $\tilde{c} = (c_1,\ldots,c_d) \in
{\rats}^d$ is the {\em rational-cost vector}, let $\mathcal{C} =
\{\tilde{x}\in {\rats}^d : \tilde{x}\geq\tilde{0},
\ \tilde{c}\cdot\tilde{x} \leq B\}\subseteq {\rats}^d$ denote the
$d$-dimensional pyramid in ${\rats}^d$ with the $d+1$ vertices
given by the origin $\tilde{0} = (0,\ldots,0)$ and
$(0,\ldots,B/c_i,\ldots,0)$, where $i\in\{1,\ldots,d\}$.  To estimate
the number of non-negative integral points in $\mathcal{C}$, we count
the number of unit $d$-cubes within the pyramid $\mathcal{C}$.  Since
$\lfloor x\rfloor \leq x \leq \lfloor x\rfloor + 1$ for each rational
$x$, then each $\tilde{x}\in\mathcal{C}$ is contained in the unit
$d$-cube with the line segment from $\lfloor\tilde{x}\rfloor =
(\lfloor x_1\rfloor, \ldots,\lfloor x_d\rfloor)$ to
$\lfloor\tilde{x}\rfloor+\tilde{1} = (\lfloor x_1\rfloor + 1,
\ldots,\lfloor x_d\rfloor+ 1)$ as its diagonal.  Since
$\tilde{c}\cdot\tilde{x} \leq B$, then
$\tilde{c}\cdot(\lfloor\tilde{x}\rfloor + \tilde{1}) \leq B +
\sum_{i=1}^dc_i$, and hence, the number of integral points in
$\mathcal{C}$ is at most the volume $\vol(\mathcal{C'})$ of the
associated pyramid $\mathcal{C'} = \{\tilde{x}\in {\rats}^d :
\tilde{x}\geq\tilde{0}, \ \tilde{c}\cdot\tilde{x} \leq B'\}
\subseteq{\rats}^d$, where $B' = B + \sum_{i=1}^dc_i$, that is,
at most $\lfloor\vol(\mathcal{C'})\rfloor$, where
\[
\vol(\mathcal{C'}) = 
\frac{1}{d!}\prod_{i=1}^d\frac{B'}{c_i}= 
\frac{1}{d!}\prod_{i=1}^d\left( 
\frac{B + \sum_{j=1}^dc_j}{c_i}\right).
\]
Note that a CSAS $T'$ of a CSM $M$ has $k_i$ edges of cost $c_i$ for
each $i$ if and only if $\tilde{k}\in\mathcal{B}\cap\mathcal{C'}$.

\begin{definition} %% XXNB! m defined formally, 14ap2015
\label{def:m}
For each $i$ let $m_i = \min(\lceil B'/c_i\rceil, n_i)$, and let 
$m = \sum_{i=1}^dm_i$.
\end{definition}
{\sc Remark:} %% XXNB! remark added 14ap2015
Note that we have 
$m = \sum_{i=1}^dm_i \leq \sum_{i=1}^dn_i = n$,
and therefore any upper bound polynomial in $m$ will yield
a bound in the same polynomial in terms of $n$.

If $\mathcal{C''} = \{0,1,\ldots,\lceil B'/c_1\rceil\} 
\times\cdots\times\{0,1,\ldots,\lceil B'/c_d\rceil\}$, then 
$\mathcal{C'}\cap {\ints}^d\subseteq\mathcal{C''}$, and
\begin{equation}
\label{eqn:k-box}
\mathcal{B}\cap\mathcal{C'} = \mathcal{B}\cap(\mathcal{C'}\cap {\ints}^d)
\subseteq  \mathcal{B}\cap\mathcal{C''}
= \{0,1,\ldots,m_1\}\times\cdots\times\{0,1,\ldots,m_d\} 
\end{equation}
Hence, by the Inequality of Arithmetic and Geometric Mean (IAGM), we
get
\[
|\mathcal{B}\cap\mathcal{C'}| \leq |\mathcal{B}\cap\mathcal{C''}|
= \prod_{i=1}^d(m_i + 1)
\leq \left(\frac{\sum_{i=1}^d(m_i+1)}{d}\right)^d = 
\left(\frac{m}{d} + 1\right)^d.
\]
We summarize in the following.
\begin{observation}
\label{obs:poly-d}
If $M$ is a CSM with $n_i$ edges of cost $c_i$ for each
$i\in\{1,\ldots,d\}$, then 
$|\mathcal{B}\cap\mathcal{C'}| \leq (m/d + 1)^d$, which is a polynomial in 
$m = \sum_{i=1}^dm_i$ of degree $d$.
\end{observation}
{\sc Remark:} Note that if $B'/c_i \leq n_i$ for each $i$, then 
$m_i = \min(\lceil B'/c_i\rceil, n_i) = \lceil B'/c_i\rceil$.  In
this case we have $\mathcal{C'}\cap {\ints}^d \subseteq \mathcal{B}$
and so $\mathcal{C'}\cap {\ints}^d 
= \mathcal{C'}\cap {\ints}^d\cap\mathcal{B} = \mathcal{C'}\cap \mathcal{B}$,
and so again by the IAGM, we obtain
\[
|\mathcal{B}\cap\mathcal{C'}| = |\mathcal{C'}\cap {\ints}^d| 
\leq \lfloor\vol(\mathcal{C'})\rfloor 
= \left\lfloor\frac{1}{d!}\prod_{i=1}^d(m_i+1)\right\rfloor 
\leq \left\lfloor\frac{1}{d!}\left(\frac{m}{d} + 1\right)^d\right\rfloor,
\]
where now $m = \sum_{i=1}^d\lceil B'/c_i\rceil$,
which shows that, although polynomial in $m$ of the same degree $d$
as in Observation~\ref{obs:poly-d},
the number of possible $\tilde{k}\in\mathcal{B}\cap\mathcal{C'}$ is a
much smaller fraction of $(m/d + 1)^d$.

%% XXNB! Geir added this summary/outline of what we do. 14ap2015
We now proceed with our setup for our dynamic programming scheme.
As before, the idea is simple; we construct a multi-dimensional matrix/array
for each vertex $u$ of $T$, the construction of which is computed in a
recursive manner, as for the previous $2\times 2$ matrices 
$\mathbf{M}(u)$ and $\mathbf{N}(u)$. 

Specifically, for each vertex $u$ we assign a 
$d(u)\times|\mathcal{B}\cap\mathcal{C'}|$-fold array
\[
\mathbf{A}(u)= \left[A_{\tilde{k}}^i(u)
\right]_{\tilde{k}\in\mathcal{B}\cap\mathcal{C'}, \ 1\leq i\leq d(u)},
\]
where $A_{\tilde{k}}^i(u)$ is the maximum prize of a subtree of $T^i(u)$
containing $k_j$ edges of cost $c_j$ for each $j\in\{1,\ldots,d\}$
and each $\tilde{k}\in\mathcal{B}\cap\mathcal{C'}$.
For $\tilde{0} = (0,\ldots,0)$, we have $A_{\tilde{0}}^i(u) = p(u)$
for each vertex $u$ for $i = 1, \ldots, d(u)$.

{\sc Convention:} For $i\in\{1,\ldots,d\}$ and an edge $e\in E(T)$,
let $\delta_i(e) = \delta_{c_i}^{c(e)}$, where for every pair of
rational numbers $x, y\in \rats$
\[
\delta_x^y = 
\left\{ 
\begin{array}{ll}
 1 &  \mbox{ if } x = y, \\
 0 &  \mbox{ otherwise }
\end{array}
\right.
\]
is the {\em Kronecker delta function}. Further, let
$\tilde{\delta}(e) = (\delta_1(e), \ldots,\delta_d(e))$.

As in (\ref{eqn:dyn-prog-rec}) and (\ref{eqn:dyn-prog-rec-C}), we use
the same decomposition of a subtree $\tau$ of $T$ into ${\tau}_{\ell}$
and $\tau''$, and as with previous Lemmas~\ref{lmm:iff-const}
and~\ref{lmm:iff-nats}, we have the following.
\begin{lemma}
\label{lmm:iff-drationals}
The subtree $\tau$ rooted at $u$ is a maximum-prize subtree among
those with $k_i$ edges of cost $c_i$ for each $i$ and that contains
the leftmost child $u_{\ell}$ of $u$ if and only if the included
subtree of $\tau_{\ell}$ is a maximum-prize subtree among those rooted
at $u_{\ell}$ and with ${\alpha}_i$ edges of cost $c_i$ for each $i$
and the included subtree of $\tau''$ is a maximum-prize subtree rooted
at $u$ among those that do not contain $u_{\ell}$ and with ${\beta}_i$
edges of cost $c_i$ for each $i$, for some $\tilde{\alpha},
\tilde{\beta}\in \mathcal{B}\cap\mathcal{C'}$, where
$\tilde{\alpha}+\tilde{\beta} =
\tilde{k}-\tilde{\delta}(e(u_{\ell}))$.
\end{lemma}
For a vertex $u$ and an arbitrary subtree $\tau$ rooted at $u$, we let
$A_{\tilde{k}}(u;\tau)$ be the maximum prize of a subtree of $\tau$
rooted at $u$ with $k_i$ edges of cost $c_i$ for each
$i\in\{1,\ldots,d\}$.  If a maximum-prize subtree of $\tau$ with $k_i$
edges of cost $c_i$ does not contain the edge from $u$ to its leftmost
child $u_{\ell}$, then $A_{\tilde{k}}(u;\tau) =
A_{\tilde{k}}(u;\tau'')$.  Otherwise, such a maximum subtree contains
${\alpha}_i$ edges of cost $c_i$ from $\tau_{\ell}$ and ${\beta}_i$
edges of cost $c_i$ from $\tau''$, where 
${\alpha}_i + {\beta}_i = c_i - \delta(e(u_{\ell}))$ 
for each $i\in\{1,\ldots,d\}$.  Finally, for
each leaf $u$ of $T$, each $i$, and 
$\tilde{k}\in \mathcal{B}\cap\mathcal{C'}$; we set 
$A_{\tilde{k}}^i(u) = p(u)$.  As
previously, we get by Lemma~\ref{lmm:iff-drationals} the following
recursion.
\begin{equation}
\label{eqn:dyn-prog-rec-drationals}
A_{\tilde{k}}(u;\tau) = \max\left( A_{\tilde{k}}(u;\tau''), 
\max_{\tilde{\alpha}+\tilde{\beta} = \tilde{k}-\tilde{\delta}(e(u_{\ell}))}
\left(A_{\tilde{\alpha}}(u_{\ell};\tau_{\ell}) 
+ A_{\tilde{\beta}}(u;\tau'')\right)\right).
\end{equation}
\begin{lemma}
\label{lmm:each-Aik}
The evaluation of each $A_{\tilde{k}}^i(u)$ takes at most 
$2(m/d + 1)^d$ arithmetic operations.
\end{lemma}
\begin{proof}
For each $\tilde{x} = (x_1,\ldots,x_d)\in {\rats}^d$, let
$\pi^+(\tilde{x}) = \prod_{i=1}^d(x_i+1)$.  By
(\ref{eqn:dyn-prog-rec-drationals}) each $A_{\tilde{k}}^i(u)$ requires
$\pi^+(\tilde{k}-\tilde{\delta}(e(u_{\ell})))$ additions and
$\pi^+(\tilde{k}-\tilde{\delta}(e(u_{\ell})))$ comparisons, and
hence all in all $2\pi^+(\tilde{k}-\tilde{\delta}(e(u_{\ell})))$
arithmetic operations.

By (\ref{eqn:k-box}) we have that $\tilde{k}\in
\mathcal{B}\cap\mathcal{C'} \subseteq\mathcal{B}\cap\mathcal{C''}$, 
and hence, $k_j\leq m_j$ for each $j\in\{1,\ldots,d\}$.  
Thus, by the IAGM, there are at most
\[
2\pi^+(\tilde{k}-\tilde{\delta}(e(u_{\ell}))) 
< 2\prod_{j=1}^d(k_j+1) 
\leq 2\prod_{j=1}^d(m_j+1)  
\leq 2\left(\frac{m}{d} + 1\right)^d
\]
arithmetic operations for evaluating each $A_{\tilde{k}}^i(u)$.
\end{proof}
Assuming each arithmetic operation takes one step, the total running
time to evaluate the entire array $\mathbf{A}(u)$ is at most a
constant multiple of
\begin{eqnarray*}
N_d(n) & = & \sum_{u\in V(T)}
\sum_{\tilde{k}\in\mathcal{B}\cap\mathcal{C'}}\sum_{i=1}^{d(u)}
2\left(\frac{m}{d} + 1\right)^d \\
       & = & \left(\sum_{u\in V(T)}d(u)\right)
\left(\sum_{\tilde{k}\in\mathcal{B}\cap\mathcal{C'}}
2\left(\frac{m}{d} + 1\right)^d \right) \\
     &\leq & 2|E(T)|\left(\frac{m}{d} + 1\right)^d
2\left(\frac{m}{d} + 1\right)^d  \\
       & = & 4(n-1)\left(\frac{m}{d} + 1\right)^{2d}.
\end{eqnarray*}
We then obtain the desired maximum prize $p(T')$ of a VCSAS $T'$ by
$p(T') = \max_{\tilde{k}\in\mathcal{B}\cap\mathcal{C'}}\left(A_{\tilde{k}}^1(r)\right)$
for the root $r$ of $T$ of our CSM $M$, which takes at most 
$|\mathcal{B}\cap\mathcal{C'}|-1<(m/d + 1)^d$ comparisons.  
Hence, we obtain the following.
\begin{theorem}
\label{thm:drationals}
If $M = (T, c, p, B, G)$ is a CSM where $T$ has $n$ vertices,
$m$ is given by Definition~\ref{def:m}, and 
$c :E(T)\rightarrow {\rats}$ takes at most $d$ distinct rational values, then
the {\sc GOAS-OP} can be solved in $O(m^{2d}n)$-time.
\end{theorem} %% XXNB! This is poly-time.!! Geir 17ap2015
{\sc Remarks:} (i) Note that when $d=1$, and hence $c_1 = c$, then $m$
in Theorem~\ref{thm:drationals} is given by 
$m = m_1 = \min(\lceil B'/c_1\rceil, n) = \min(\lceil B/c\rceil + 1, n)$, 
whereas in Theorem~\ref{thm:c=const} 
$m = \lceil B/c\rceil = \min(\lceil B/c\rceil, n)$, by the assumption that 
$\lceil B/c\rceil\leq n$.  Still, the complexity when $d=1$ in 
Theorem~\ref{thm:drationals} clearly agrees with the complexity of 
$O(m^2n)$ for solving the {\sc GOAS-OP} when $c$ is a constant function in
Theorem~\ref{thm:c=const}.  (ii) If each $m_i = O(f(n))$, for some
``slow-growing'' function of $n$, then Theorem~\ref{thm:drationals} yields
an $O(nf(n)^{2d})$-time algorithm for solving the {\sc GOAS-OP}\@.  In
particular, if each $m_i = O(1)$, then Theorem~\ref{thm:drationals} yields
a linear-time in $n$ algorithm to solve the {\sc GOAS-OP}.
\begin{corollary}
The {\sc GOAS-DP} when restricted to $d$ rational-valued penetration costs
can be solved in polynomial time. %% XXNB! poly, not pseudo poly.! Geir 17ap2015
\end{corollary}

\section{Summary and Conclusions}
\label{sec:summary}

This paper defined a new cyber-security model that models
systems which are designed based on defense-in-depth.  We showed that
natural problems based on the model were intractable.  We then proved
that restricted versions of the problems had either polynomial time or
pseudo-polynomial time algorithms. Table~\ref{tab:summary} 
in Section~\ref{sec:introduction} summarizes our results.
They suggest that in a real system the penetration costs
should vary, that is, although each level should be difficult to
attack, the cost of breaking into some levels should be even higher.
The tree representation of the models suggests that systems should be
designed to distribute targets in a bushy tree, rather than in a
narrow tree.  Most security systems are linear, and such systems could
be strengthen by distributing targets more widely, providing
{\em defense-in-deception}.  Although in most
situations a cyber attacker will not a priori know exact penetration
costs, target locations, and prizes, the model still gives us insight
into which types of security designs would be more effective.

We conclude the paper with a number of open questions.
\begin{enumerate}
\item
Can we quantify how much targets need to be distributed in order to
maximize security?  For example, does an $(n+1)$-ary tree provide
provably better security than an $n$-ary tree?
\item
Can we prove mathematically that the intuition of storing high-value
targets deeper in the system and having higher penetration costs on the
outer-most layers of the system results in the best security?
\item
If targets are allowed to be repositioned periodically, what does
that do to the complexity of the problems, and what is the
best movement strategy for protecting targets?
\item
Using the model, can one develop a set of benchmarks to rank the
security of a particular system?  How would one model prizes in a system?
\item
Can the notion of time and intrusion detection be built into the
model?  That is, if an attacker tries to break into a certain
container, the attacker may be locked out, resulting in game-over for
that attacker, or perhaps may face an even higher new penetration
cost.
\item
Are there online variants of the model that are interesting to study?
For example, a version where the topology of the graph changes
dynamically or where only a partial description is known to the
attacker.
\end{enumerate}

\section*{Acknowledgments}

This work was in part motivated by a talk that Bill Neugent of MITRE
Corporation gave at the United States Naval Academy in the fall of
2011.  We thank Bill for initial discussions about game-over issues
relating to cyber-security models.  Thanks also to Richard Chang for
discussions about the model. -- Finally, we like to thank the two
anonymous referees for their careful reading of the paper, their
pointed comments and suggestions which resulted in a greatly improved
presentation of the results and made them more complete.

\end{document}